\documentclass[lettersize,journal]{IEEEtran}
\usepackage{amsmath,amsfonts}
\usepackage{amsmath,amssymb,amsthm}

\newtheorem{theorem}{Theorem}
\newtheorem{proposition}{Proposition}

\newtheorem{corollary}{Corollary}

\usepackage{algorithmic}
\usepackage{algorithm}
\usepackage{array}
\usepackage[caption=false,font=normalsize,labelfont=sf,textfont=sf]{subfig}
\usepackage{textcomp}
\usepackage{stfloats}
\usepackage{url}
\usepackage{verbatim}
\usepackage{booktabs}
\usepackage{siunitx}
\usepackage{amssymb}
\usepackage{graphicx}
\usepackage{multirow}
\def\BibTeX{{\rm B\kern-.05em{\sc i\kern-.025em b}\kern-.08em
    T\kern-.1667em\lower.7ex\hbox{E}\kern-.125emX}}
\usepackage{balance}
\begin{document}
\title{A Hypothesis-Testing Analysis of Blind Recognition for Polar Codes}
\author{Changwei Tu, Cheng Yang, Xianzhao Feng, and Kai Niu
\thanks{This work is supported by the National Natural Science Foundation of China under Grant 62321001 and Grant 62471054.
  \emph{(Corresponding author: Kai Niu.)}
  
The authors are with the Key Laboratory of Universal Wireless Communications, Ministry of Education,
Beijing University of Posts and Telecommunications, Beijing, China.
Email: \{tuchangwei, yscheng, foreseexz, niukai\}@bupt.edu.cn.}
  }

\markboth{Journal of \LaTeX\ Class Files,~Vol.~18, No.~9, September~2020}%
{How to Use the IEEEtran \LaTeX \ Templates}

\maketitle

\begin{abstract}
Blind recognition of polar-coded transmissions is an important task in non-cooperative wireless forensics and security-oriented signal analysis.
When the code length is known or has been estimated, recovering the frozen/information bit-position pattern is a key step in identifying the underlying polar-code structure and enabling subsequent information recovery from intercepted observations.
In this paper, blind recognition of polar codes is investigated from a hypothesis-testing perspective under the successive cancellation (SC)-based synthetic bit-channel representation. 
First, under an ideal SC-consistent condition, we formulate position-wise recognition as a binary hypothesis test between frozen-position and information-position models, which provides a theoretical benchmark for analyzing their intrinsic distinguishability.
Second, we show that the adopted soft recognition metric admits an exact shifted log-likelihood-ratio interpretation. This justifies $\ln 2$ as the neutral threshold under equal priors and costs, while unequal priors or costs lead to the corresponding Bayesian threshold shift.
Third, under the ideal SC-consistent model and this neutral setting, we derive upper and lower bounds on the position-wise and sequence-level recognition error probabilities with multiple independent observations. The resulting overlap coefficient is further related to the classical Bhattacharyya parameter, establishing an interpretable link between blind-recognition difficulty and polar synthetic-channel reliability.
Simulation results show that the derived bounds characterize the recognition performance under the ideal SC-consistent model and capture the effects of code length, the number of intercepted observations, and SNR.
Further paired comparisons in the tested settings indicate that the SC-consistent recursion provides a good sequence-level match to the realistic SC-recursive procedure.
\end{abstract}


\section{Introduction}

\IEEEPARstart{I}{n} non-cooperative wireless environments, a receiver may intercept and analyze transmissions without protocol assistance, signaling metadata, or prior knowledge of the transmitter configuration \cite{Dobre2007,Shiu2011}. Such scenarios commonly arise in spectrum monitoring, wireless forensics, electronic surveillance, and security-oriented signal analysis \cite{MoosaviLarsson2011,Xu2011}. 
As a result, key physical-layer parameters must be inferred directly from noisy observations rather than obtained from cooperative signaling. Accurate recovery of the underlying transmission structure is therefore essential for reliable downstream processing, including protocol interpretation and forensic assessment.

In this context, channel-coding information is of particular importance because it reflects block-level constraints beyond front-end descriptors. 
Unlike modulation format and scrambling structure, channel codes impose global algebraic and statistical dependencies on received sequences, directly affecting reliable decoding and information recovery \cite{MoosaviLarsson2014,Dehdashtian2021}. 
Accordingly, blind channel-code recognition provides channel-coding evidence for protocol reverse analysis and is relevant to the security-oriented assessment of wireless transmissions \cite{Mukherjee2014,PoorSchaefer2017}.

Polar codes, originally introduced by Ar{\i}kan, are the first class of channel codes that are provably capacity-achieving for binary-input memoryless symmetric channels \cite{Arikan2009}. Owing to their recursive structure and favorable performance--complexity tradeoff, polar codes have been extensively studied in modern communication systems \cite{TalVardy2015,NiuChen2012}. In particular, their adoption in 5G NR control-channel coding \cite{TS38212} makes polar-coded transmissions an important object of non-cooperative wireless forensics and security-oriented signal analysis. Blind detection of polar-coded control messages has been investigated in 5G-related settings, where the receiver may need to analyze polar-coded transmissions without transmitter-side configuration information \cite{Condo2017}. Beyond detection-level analysis, a more detailed structural-recognition problem is to recover the internal coding constraints of polar codes from limited intercepted observations. Once the code length is known or has been estimated, the frozen/information position pattern becomes a key structural parameter because it determines the actual polar code used by the transmitter.

Existing studies on blind recognition of polar codes can be roughly classified according to the type of information used for recognition. The first group relies on hard decisions, thresholding rules, matrix structures, or recursive inheritance properties to distinguish frozen and information positions \cite{Liu2022,Yi2023,Xu2026}. The second group exploits soft observations and reliability information extracted from the received signal to improve robustness under noisy conditions \cite{Wang2023}. These methods have shown that frozen/information position patterns can be inferred from non-cooperative observations, but their performance may still degrade when the observations are limited, the channel condition is unfavorable, or the recognition metric is not well matched to the underlying statistical model. More recently, \cite{TuBSCL} combines blind recognition with successive cancellation list (SCL) decoding, so that soft information generated during the decoding process can be exploited more effectively to improve the recognition success rate.

Among the above studies, the SCL-aided framework in \cite{TuBSCL} is the closest to the present work, because its recognition rule separates frozen and information positions by using a soft metric together with a reference threshold. Although this rule improves empirical recognition performance, the statistical roles of the metric and the threshold have not been fully characterized. In particular, it remains unclear how the adopted metric should be interpreted as evidence in a position-wise frozen/information test, why the reference threshold is justified under the symmetric equal-prior and equal-cost setting, and how the resulting recognition reliability depends on the number of intercepted observations and the reliability of the corresponding polar synthetic channels.

Hypothesis testing provides a classical framework for distinguishing competing probabilistic models from observations~\cite{VanTrees2001}. 
It is therefore suited to the present problem, where frozen and information positions induce different likelihood models under the SC-based synthetic bit-channel representation. 
In this paper, we analyze polar-code blind recognition from this viewpoint. 
Conditioned on the ideal SC-consistent event, each position-wise soft-decision problem is formulated as a binary test between a frozen-position model and an information-position model. 
We prove that the soft recognition metric used in \cite{TuBSCL} is the corresponding log-likelihood ratio shifted by $\ln 2$. 
This result explains the threshold $\ln 2$ for the averaged metric under the equal-prior and equal-cost local Bayesian rule, while unequal local priors or costs lead to the corresponding Bayesian threshold shift. 
We further derive position-wise and sequence-level upper and lower bounds for multiple independent intercepted observations under the neutral rule and the ideal SC-consistent model. 
The bounds are governed by the overlap coefficient between the two position-type distributions, which is further related to the classical Bhattacharyya parameter of polar synthetic channels~\cite{Arikan2009}.

The main contributions of this paper are summarized as follows.
\begin{enumerate}

    \item 
          We establish a hypothesis-testing framework for blind recognition of polar codes under the SC-based synthetic bit-channel representation. Under an ideal SC-consistent condition, the recognition of each source position is formulated as a binary statistical decision problem between a frozen-position model and an information-position model, which provides a theoretical benchmark for characterizing the intrinsic distinguishability of frozen and information positions.

     \item 
          Within this framework, we establish a decision-theoretic interpretation of the soft recognition rule. Specifically, under the ideal SC-consistent synthetic-channel model, we prove that the adopted soft metric is exactly a shifted log-likelihood ratio for discriminating between the frozen-position and information-position hypotheses. Consequently, the reference threshold $\ln 2$ arises naturally under equal priors and costs, while unequal priors or costs induce the corresponding threshold shift.

    \item 
         Building on this neutral setting, we derive position-wise and sequence-level upper and lower bounds for multiple independent intercepted observations under the ideal SC-consistent model. We further relate the resulting overlap coefficient to the classical Bhattacharyya parameter of polar synthetic channels, thereby linking blind-recognition reliability to conventional polar-code reliability analysis.

\end{enumerate}

The remainder of this paper is organized as follows.
Section~II introduces the system model, SC decoding for polar codes, and binary hypothesis testing.
Section~III formulates blind recognition as a hypothesis-testing problem under the SC-based synthetic bit-channel representation and presents the statistical interpretation of the recognition metric and its reference threshold.
Section~IV develops the corresponding performance analysis, including upper and lower bounds under the neutral symmetric formulation, and establishes their connection to the classical Bhattacharyya parameter of polar synthetic channels.
Section~V provides simulation results to support the proposed analysis.
Finally, Section~VI concludes this paper.

\section{Preliminaries}
\subsection{System Model}

Let $\mathcal{C}(N,K)$ denote a polar code of length $N=2^n$ and dimension $K$. 
After channel polarization, the $K$ most reliable synthetic subchannels are assigned to information bits, while the remaining $N-K$ subchannels are assigned to frozen bits. 
Let $\mathcal{I}$ and $\mathcal{F}$ denote the information set and frozen set, respectively. 
For the source vector $\mathbf{u}=[u_0,u_1,\ldots,u_{N-1}]$, with $u_i=0$ for $i\in\mathcal{F}$, the corresponding codeword is generated as
\begin{equation}
\mathbf{c}=\mathbf{u}\mathbf{G}_N,
\end{equation}
where
\begin{equation*}
\mathbf{G}_N=\begin{bmatrix}1&0\\1&1\end{bmatrix}^{\otimes n}
\end{equation*}
is the generator matrix.

In this work, the code length $N$ is assumed known at the receiver, whereas the dimension $K$ and the frozen/information position pattern are unknown.
Suppose that $M$ independent codewords, each of length $N$, are transmitted over an additive white Gaussian noise (AWGN) channel using binary phase-shift keying (BPSK) modulation. 
Then, the received sample corresponding to the $j$-th coded bit of the $m$-th codeword is
\begin{equation*}
r_{m,j}=x_{m,j}+n_{m,j},
\qquad
0\le m\le M-1,\;
0\le j\le N-1,
\end{equation*}
where $x_{m,j}=1-2c_{m,j}\in\{+1,-1\}$ is the BPSK-modulated symbol corresponding to $c_{m,j}$, and $n_{m,j}\sim\mathcal{N}(0,\sigma^2)$ denotes the Gaussian noise sample. 
The corresponding channel log-likelihood ratio (LLR) is
\begin{equation}
\lambda_{m,j}
=
\ln
\frac{p(r_{m,j}|c_{m,j}=0)}
     {p(r_{m,j}|c_{m,j}=1)}
=
\frac{2r_{m,j}}{\sigma^2}.
\end{equation}
Collecting the channel LLRs of all $M$ received vectors yields
\begin{equation}
\boldsymbol{\Lambda}^{\rm ch}=
\begin{bmatrix}
\lambda_{0,0} & \lambda_{0,1} & \cdots & \lambda_{0,N-1}\\
\lambda_{1,0} & \lambda_{1,1} & \cdots & \lambda_{1,N-1}\\
\vdots & \vdots & \ddots & \vdots\\
\lambda_{M-1,0} & \lambda_{M-1,1} & \cdots & \lambda_{M-1,N-1}
\end{bmatrix},
\end{equation}
which serves as the input for subsequent blind recognition.

\subsection{SC Decoding of Polar Codes}

For a polar code of length $N$, SC decoding proceeds recursively over the decoding tree. For notational simplicity, we consider one received vector and denote by $\lambda_j$ the channel LLR associated with the received sample $r_j$. For a length-$2$ polar code, the synthetic LLR associated with $u_0$ is first computed by the $f$-operation as
\begin{equation}
\Lambda_0
=
f(\lambda_0,\lambda_1)
=
\lambda_0\boxplus\lambda_1,
\end{equation}
where
\begin{equation*}
a\boxplus b
=
2\operatorname{arctanh}\!\left(\tanh\frac{a}{2}\tanh\frac{b}{2}\right).
\end{equation*}
After obtaining $\hat{u}_0$ by hard decision, the synthetic LLR associated with $u_1$ is computed through the $g$-operation as
\begin{equation}
\Lambda_1
=
g(\lambda_0,\lambda_1,\hat{u}_0)
=
(-1)^{\hat{u}_0}\lambda_0+\lambda_1 .
\end{equation}

For a general polar code, the same $f$- and $g$-operations are recursively applied over the decoding tree. In this way, SC decoding successively produces the synthetic LLRs
\begin{equation*}
\Lambda_0,\Lambda_1,\ldots,\Lambda_{N-1},
\end{equation*}
where $\Lambda_i$ denotes the LLR associated with the $i$-th source bit $u_i$ under the SC recursion.

When the frozen set is known, the standard SC hard-decision rule is
\begin{equation}
\hat{u}_i=
\begin{cases}
0, & i\in\mathcal{F},\\
0, & i\in\mathcal{I},\ \Lambda_i\ge 0,\\
1, & i\in\mathcal{I},\ \Lambda_i<0.
\end{cases}
\end{equation}
By recursively applying the $f$- and $g$-operations, SC decoding outputs the estimated source vector
\[
\hat{\mathbf{u}}=[\hat{u}_0,\hat{u}_1,\ldots,\hat{u}_{N-1}].
\]

The above recursion also provides the synthetic bit-channel representation used in polar-code analysis. Specifically, the $i$-th SC synthetic channel is denoted by $W_i(\cdot|\cdot)$. Following the conventional notation for polar synthetic channels, let $y$ denote the generic output of this synthetic channel. The corresponding synthetic LLR is defined as
\begin{equation}
\Lambda_i(y)
=
\ln\frac{W_i(y|0)}{W_i(y|1)} .
\end{equation}
This representation will be used in the following sections to formulate the blind recognition of frozen and information positions.

\subsection{Binary Hypothesis Testing} 
In a binary hypothesis-testing problem, an observation $Y$ is assumed to be generated according to one of two probability laws corresponding to hypotheses $\mathcal{H}_0$ and $\mathcal{H}_1$, namely,
\begin{equation}
\mathcal{H}_0: Y\sim P,
\qquad
\mathcal{H}_1: Y\sim Q .
\end{equation}

Throughout this paper, when a probability law is evaluated at an observation, it denotes the corresponding density or probability mass function with respect to the underlying common dominating measure. 
Thus, $P(y)$ and $Q(y)$ denote the likelihoods of the observation $y$ under $\mathcal{H}_0$ and $\mathcal{H}_1$, respectively. 
The log-likelihood ratio for discriminating $\mathcal{H}_1$ from $\mathcal{H}_0$ is defined as
\begin{equation}
\Gamma(Y)\triangleq \ln \frac{Q(Y)}{P(Y)} .
\end{equation}

Under the Bayesian criterion, let $\pi_0$ and $\pi_1$ denote the prior probabilities of $\mathcal{H}_0$ and $\mathcal{H}_1$, respectively, and let $c_{0\to 1}$ and $c_{1\to 0}$ denote the corresponding misclassification costs. 
For a given observation $Y$, deciding $\mathcal{H}_1$ incurs the conditional risk proportional to $\pi_0 c_{0\to 1}P(Y)$, whereas deciding $\mathcal{H}_0$ incurs the conditional risk proportional to $\pi_1 c_{1\to 0}Q(Y)$. 
Therefore, the optimal Bayesian decision rule is
\begin{equation}
\Gamma(Y)
\mathop{\gtrless}_{\mathcal{H}_0}^{\mathcal{H}_1}
\tau,
\qquad
\tau \triangleq \ln \frac{\pi_0 c_{0\to 1}}{\pi_1 c_{1\to 0}} .
\end{equation}
In particular, when the two hypotheses are equiprobable and the two types of decision errors have the same cost, the threshold reduces to $\tau=0$.

\section{Hypothesis-Testing Analysis of Blind Recognition}

In this section, we specialize the binary hypothesis-testing framework in Section~II to blind recognition of polar codes under an ideal SC-consistent setting. 
The goal is to clarify the statistical meaning of the recognition metric and to derive the corresponding Bayesian decision threshold. 

\subsection{Log-Likelihood-Ratio Interpretation of the Recognition Metric}

For the $i$-th source-bit decision under SC decoding, the SC recursion
depends on the preceding decisions $\hat{u}_0^{i-1}$. Accordingly, define
\begin{equation}
\mathcal{A}_i \triangleq
\left\{\hat{u}_0^{i-1}=u_0^{i-1}\right\},
\label{eq:Ai}
\end{equation}
as the event that all decisions associated with source-bit positions
$0,\ldots,i-1$ are correct.

This event represents the SC-consistent condition under which the synthetic LLR associated with the $i$-th position follows the standard polar synthetic-channel representation. 
Conditioned on this ideal event, the position-type decision at index $i$ can be modeled as a binary hypothesis test between
\begin{align*}
\mathcal{H}_{F,i}&:\ \text{the $i$-th position is a frozen position},\\
\mathcal{H}_{I,i}&:\ \text{the $i$-th position is an information position}.
\end{align*}

Using the synthetic-channel notation introduced in Section~II-B, the source bit is fixed to $0$ under the frozen-position hypothesis. 
In contrast, under the information-position hypothesis, the source bit is modeled as equiprobable, i.e., $u_i\sim\mathrm{Bernoulli}(1/2)$. 
Therefore,  the corresponding single-observation probability laws are characterized by the likelihoods
\begin{equation*}
P_i(y) \triangleq W_i(y|0),
\qquad
Q_i(y) \triangleq \frac{1}{2}W_i(y|0)+\frac{1}{2}W_i(y|1),
\end{equation*}
where $P_i$ corresponds to the frozen-position model and $Q_i$ corresponds to the information-position model. 
Consistent with the convention introduced in Section~II-C, $P_i(y)$ and $Q_i(y)$ denote the likelihoods of the synthetic observation $y$ under $\mathcal{H}_{F,i}$ and $\mathcal{H}_{I,i}$, respectively.

Using the synthetic LLR $\Lambda_i(y)$ defined in Section~II-B, and motivated by the frozen-bit metric increment used under the frozen-bit hypothesis in~\cite{TuBSCL}, define the single-observation recognition metric as
\begin{equation}
C_i(y) \triangleq \ln\!\bigl(1+e^{-\Lambda_i(y)}\bigr).
\end{equation}
According to the binary hypothesis-testing formulation in Section~II-C, the log-likelihood ratio for discriminating the information-position hypothesis from the frozen-position hypothesis is
\begin{equation}
\Gamma_i(y) \triangleq \ln \frac{Q_i(y)}{P_i(y)}.
\end{equation}
The following theorem shows that $C_i(y)$ is exactly a constant-shifted version of this position-type log-likelihood ratio under the SC-consistent condition.

\begin{theorem}
\label{thm:metric_relation}
Under $\mathcal{A}_i$, the single-observation recognition metric satisfies
\begin{equation}
C_i(y)=\Gamma_i(y)+\ln 2.
\end{equation}
\end{theorem}

\begin{proof}
From the definition of $\Lambda_i(y)$ in Section~II-B, we have
\[
e^{-\Lambda_i(y)}=\frac{W_i(y|1)}{W_i(y|0)}.
\]
Hence,
\begin{align*}
C_i(y)
&= \ln\!\left(1+e^{-\Lambda_i(y)}\right) \\
&= \ln\!\left(1+\frac{W_i(y|1)}{W_i(y|0)}\right) \\
&= \ln\!\left(\frac{W_i(y|0)+W_i(y|1)}{W_i(y|0)}\right).
\end{align*}
Under $\mathcal{A}_i$, the definitions of $P_i(y)$ and $Q_i(y)$ imply that
\[
P_i(y)=W_i(y|0),
\qquad
2Q_i(y)=W_i(y|0)+W_i(y|1).
\]
Substituting these identities into the above expression yields
\begin{align*}
C_i(y)
&= \ln\!\left(\frac{2Q_i(y)}{P_i(y)}\right) \\
&= \ln \frac{Q_i(y)}{P_i(y)}+\ln 2 \\
&= \Gamma_i(y)+\ln 2.
\end{align*}
This completes the proof.
\end{proof}

Theorem~\ref{thm:metric_relation} shows that the recognition metric admits an exact shifted log-likelihood-ratio interpretation for distinguishing the frozen-position model from the information-position model under the SC-consistent setting. 
Under this interpretation, $\ln 2$ appears as the constant shift between the recognition metric and the corresponding position-type log-likelihood ratio. 
This constant shift is induced by the information-position model and the definition of the soft metric itself, and is independent of the prior probabilities assigned to the two position-type hypotheses. 
This interpretation is also compatible with the use of $\ln 2$ in~\cite{TuBSCL}, where $\ln 2$ serves as a neutral reference metric for the information-bit hypothesis.

Now suppose that $M$ independent intercepted observations are available for the $i$-th position. 
Let $Y_1,\ldots,Y_M$ denote the corresponding outputs of the $i$-th synthetic bit-channel induced by these observations under the same SC-consistent condition. 
Define the accumulated statistics
\begin{equation*}
S_i^{(M)} \triangleq \sum_{m=1}^{M} C_i(Y_m),
\qquad
G_i^{(M)} \triangleq \sum_{m=1}^{M} \Gamma_i(Y_m).
\end{equation*}
Then, by Theorem~\ref{thm:metric_relation},
\begin{equation}
S_i^{(M)} = G_i^{(M)} + M\ln 2.
\end{equation}
Therefore, the accumulated recognition metric can be understood as a standard multi-sample log-likelihood-ratio statistic after a constant shift. 
This equivalence provides the basis for the Bayesian threshold derived in the next subsection.

\subsection{Bayesian Threshold for the Averaged Recognition Metric}

We next derive the Bayesian threshold of the averaged recognition metric and then specialize it to the neutral symmetric setting.
Under $\mathcal{A}_i$, let $\pi_{F,i} \triangleq \Pr(\mathcal{H}_{F,i})$ and $\pi_{I,i} \triangleq \Pr(\mathcal{H}_{I,i})$ denote the local prior probabilities of the frozen-position and information-position hypotheses, respectively.
Let $c_{F\to I}$ and $c_{I\to F}$ denote the corresponding misclassification costs.
By the general Bayesian decision rule introduced in Section~II, the optimal local test based on $\Gamma_i(y)$ is
\begin{equation}
\Gamma_i(y)
\mathop{\gtrless}_{\mathcal{H}_{F,i}}^{\mathcal{H}_{I,i}}
\tau_i,
\qquad
\tau_i \triangleq \ln \frac{\pi_{F,i}c_{F\to I}}{\pi_{I,i}c_{I\to F}} .
\label{eq:bayes_rule_single}
\end{equation}

For $M$ independent observations, the accumulated log-likelihood ratio is
$G_i^{(M)}=\sum_{m=1}^{M}\Gamma_i(Y_m)$.
Since the prior probabilities and misclassification costs are assigned to the whole $M$-sample local decision problem, rather than to each individual observation, the corresponding Bayesian decision rule becomes
\begin{equation}
G_i^{(M)}
\mathop{\gtrless}_{\mathcal{H}_{F,i}}^{\mathcal{H}_{I,i}}
\tau_i .
\label{eq:bayes_rule_multi}
\end{equation}
Using the relation $S_i^{(M)}=G_i^{(M)}+M\ln 2$, the above rule can be equivalently rewritten in terms of the averaged recognition metric as
\begin{equation}
\frac{1}{M}S_i^{(M)}
\mathop{\gtrless}_{\mathcal{H}_{F,i}}^{\mathcal{H}_{I,i}}
\ln 2+\frac{\tau_i}{M}.
\label{eq:threshold_metric_multi_general}
\end{equation}
Thus, $\ln 2$ comes from the constant shift between the recognition metric and the log-likelihood ratio, whereas $\tau_i/M$ accounts for local prior or cost asymmetry.

In the non-cooperative scenario, no prior information is available about the underlying hypothesis. Therefore, we consider a neutral symmetric setting, in which the two hypotheses are assigned equal local priors and the two types of decision errors are assigned equal costs, i.e.,
\[
\pi_{F,i}=\pi_{I,i},
\qquad
c_{F\to I}=c_{I\to F}.
\]
Hence, $\tau_i=0$.
For a single observation, the optimal decision rule based on $\Gamma_i(y)$ reduces to
\[
\Gamma_i(y)
\mathop{\gtrless}_{\mathcal{H}_{F,i}}^{\mathcal{H}_{I,i}}
0 .
\]
By Theorem~\ref{thm:metric_relation}, this rule is equivalent to
\[
C_i(y)
\mathop{\gtrless}_{\mathcal{H}_{F,i}}^{\mathcal{H}_{I,i}}
\ln 2 .
\]
Similarly, for $M$ independent observations, \eqref{eq:threshold_metric_multi_general} reduces to
\begin{equation}
\frac{1}{M}S_i^{(M)}
\mathop{\gtrless}_{\mathcal{H}_{F,i}}^{\mathcal{H}_{I,i}}
\ln 2 .
\label{eq:threshold_metric_multi}
\end{equation}
Therefore, under the neutral equal-prior and equal-cost local decision rule adopted in this paper, the averaged recognition metric is compared with the threshold $\ln 2$.

Here, the equal-prior assumption is used only as a neutral local prior for each position-type test.
It does not imply that the actual numbers of frozen and information positions are equal, nor does it impose any global cardinality constraint on the information set.
In the following analysis and simulations, we adopt this neutral local rule.

To further illustrate the role of this neutral reference threshold, define
\[
\mu_{F,i} \triangleq \mathbb{E}_{P_i}[C_i(Y)],
\qquad
\mu_{I,i} \triangleq \mathbb{E}_{Q_i}[C_i(Y)].
\]

\begin{corollary}
\label{prop:mean_separation}
Under $\mathcal{A}_i$, the means of the recognition metric under the two hypotheses satisfy
\begin{equation}
\mu_{F,i}=\ln 2-D(P_i\|Q_i),
\qquad
\mu_{I,i}=\ln 2+D(Q_i\|P_i),
\label{eq:mean_relation}
\end{equation}
where $D(\cdot\|\cdot)$ denotes the Kullback--Leibler divergence~\cite{CoverThomas2006}.
Consequently, if $P_i\neq Q_i$, then
\begin{equation}
\mu_{F,i}<\ln 2<\mu_{I,i}.
\label{eq:mean_separation}
\end{equation}
\end{corollary}

\begin{proof}
By Theorem~\ref{thm:metric_relation},
\[
C_i(Y)=\Gamma_i(Y)+\ln 2
      =\ln \frac{Q_i(Y)}{P_i(Y)}+\ln 2 .
\]
Taking expectation under $P_i$ gives
\begin{align*}
\mu_{F,i}
&= \mathbb{E}_{P_i}\!\left[\ln \frac{Q_i(Y)}{P_i(Y)}\right] + \ln 2 \\
&= -D(P_i\|Q_i)+\ln 2 .
\end{align*}
Similarly, taking expectation under $Q_i$ yields
\begin{align*}
\mu_{I,i}
&= \mathbb{E}_{Q_i}\!\left[\ln \frac{Q_i(Y)}{P_i(Y)}\right] + \ln 2 \\
&= D(Q_i\|P_i)+\ln 2 .
\end{align*}
Since the Kullback--Leibler divergence is nonnegative and equals zero if and only if the two distributions are identical, we have
\[
\mu_{F,i}\le \ln 2,
\qquad
\mu_{I,i}\ge \ln 2,
\]
with strict inequalities whenever $P_i\neq Q_i$.
This completes the proof.
\end{proof}

Corollary~1 shows that, under the SC-consistent model, the expected value of the recognition metric lies below $\ln 2$ for the frozen-position model and above $\ln 2$ for the information-position model.
Hence, $\ln 2$ lies between the two hypothesis-dependent means.

Under the neutral local rule, the averaged recognition metric is therefore compared with $\ln 2$.
The resulting position-wise and sequence-level error probabilities are analyzed next.

\section{Performance Analysis of Blind Recognition}

In this section, we analyze the recognition error probability under the ideal SC-consistent hypothesis-testing framework established in Section~III.
We first derive position-wise and sequence-level upper and lower bounds for the recognition problem under the neutral symmetric formulation, and then relate the resulting overlap coefficient to the classical Bhattacharyya parameter of polar synthetic channels.

\subsection{Hypothesis-Testing-Based Upper Bounds on Recognition Error Probability}

We first derive an upper bound on the conditional position-wise recognition error probability. 
As in Section~III, the analysis at position $i$ is conditioned on the SC-consistent event $\mathcal{A}_i$, under which the synthetic LLR follows the standard polar synthetic-channel representation. 
Under the neutral symmetric decision rule, the position-type decision at the $i$-th source index is made by comparing the accumulated log-likelihood ratio $G_i^{(M)}$ with zero.

To this end, define the Bhattacharyya coefficient between two probability measures $P$ and $Q$ as
\begin{equation}
B(P,Q)
\triangleq
\int \sqrt{p(y)q(y)}\,dy,
\label{eq:general_bhattacharyya}
\end{equation}
where $p$ and $q$ denote the corresponding densities with respect to a common dominating measure. 
In the following, when no confusion arises, the same symbols are used for probability measures and their corresponding densities.

The Bhattacharyya coefficient associated with the position-type discrimination problem is defined as
\begin{equation}
\Omega_i
\triangleq
B(P_i,Q_i)
=
\int \sqrt{P_i(y)Q_i(y)}\,dy .
\label{eq:bhat_coeff}
\end{equation}
This coefficient measures the overlap between the frozen-position distribution $P_i$ and the information-position distribution $Q_i$.

Conditioned on $\mathcal{A}_i$ and under the neutral symmetric decision rule, define the two one-sided error probabilities at position $i$ by
\begin{align}
P_{F\to I,i}^{(M)}
&\triangleq
\Pr\nolimits_{P_i^{\bullet M}}
\left\{G_i^{(M)}>0\right\}, \notag\\
P_{I\to F,i}^{(M)}
&\triangleq
\Pr\nolimits_{Q_i^{\bullet M}}
\left\{G_i^{(M)}\le 0\right\},
\label{eq:one_sided_errors}
\end{align}
where $P_i^{\bullet M}$ and $Q_i^{\bullet M}$ denote the $M$-fold product measures induced by $P_i$ and $Q_i$, respectively.

Let $\mathcal{E}_i$ denote the decision-error event of the local position-type test at source index $i$. 
Under the neutral symmetric local Bayesian formulation, this event can be written as
\begin{equation}
\mathcal{E}_i \triangleq
\bigl(\mathcal{H}_{F,i}\!\cap\!\{G_i^{(M)}>0\}\bigr)
\!\cup\!
\bigl(\mathcal{H}_{I,i}\!\cap\!\{G_i^{(M)}\le 0\}\bigr).
\label{eq:local_error_event}
\end{equation}
The first term corresponds to a frozen-to-information error, while the second term corresponds to an information-to-frozen error. 
Under the equal-prior and equal-cost setting, the corresponding local recognition error probability is defined as
\begin{equation}
P_{e,i}^{(M)}
\triangleq
\Pr(\mathcal{E}_i\mid\mathcal{A}_i)
=
\frac{1}{2}P_{F\to I,i}^{(M)}
+
\frac{1}{2}P_{I\to F,i}^{(M)} .
\label{eq:position_error_probability}
\end{equation}

\begin{theorem}
\label{thm:upper_position}
Under $\mathcal{A}_i$ and the neutral symmetric decision rule, the one-sided conditional error probabilities satisfy
\begin{equation}
P_{F\to I,i}^{(M)} \le \Omega_i^M,
\qquad
P_{I\to F,i}^{(M)} \le \Omega_i^M .
\label{eq:upper_position_one_sided}
\end{equation}
Consequently,
\begin{equation}
P_{e,i}^{(M)} \le \Omega_i^M .
\label{eq:upper_position_symmetric}
\end{equation}
\end{theorem}

\begin{proof}
We first consider the frozen-to-information error event. 
For any $s>0$, Markov's inequality gives
\begin{align*}
P_{F\to I,i}^{(M)}
&=
\Pr\nolimits_{P_i^{\bullet M}}
\left\{G_i^{(M)}>0\right\} \\
&=
\Pr\nolimits_{P_i^{\bullet M}}
\left\{e^{sG_i^{(M)}}>1\right\} \\
&\le
\mathbb{E}_{P_i^{\bullet M}}
\left[e^{sG_i^{(M)}}\right].
\end{align*}
Since the $M$ observations are independent,
\[
\mathbb{E}_{P_i^{\bullet M}}
\left[e^{sG_i^{(M)}}\right]
=
\left(
\mathbb{E}_{P_i}
\left[e^{s\Gamma_i(Y)}\right]
\right)^M .
\]
Using the definition of $\Gamma_i(y)$, we obtain
\begin{align*}
\mathbb{E}_{P_i}
\left[e^{s\Gamma_i(Y)}\right]
&=
\int P_i(y)
\left(\frac{Q_i(y)}{P_i(y)}\right)^s dy \\
&=
\int P_i(y)^{1-s}Q_i(y)^s\,dy .
\end{align*}
Here, $s$ plays the role of a Chernoff parameter. 
Optimizing over $s$ could yield a tighter Chernoff-type bound. 
In this work, we choose $s=\frac{1}{2}$ so that the resulting bound is expressed in terms of the position-type Bhattacharyya coefficient $\Omega_i$ defined in \eqref{eq:bhat_coeff}. 
This choice is also consistent with the subsequent comparison between $\Omega_i$ and the classical Bhattacharyya parameter of the polar synthetic channel. 
With $s=\frac{1}{2}$, we have
\[
\mathbb{E}_{P_i}
\left[e^{\frac{1}{2}\Gamma_i(Y)}\right]
=
\int \sqrt{P_i(y)Q_i(y)}\,dy
=
\Omega_i ,
\]
and hence
\[
P_{F\to I,i}^{(M)} \le \Omega_i^M .
\]

Next, consider the information-to-frozen error event. 
Again by Markov's inequality, for any $s>0$,
\begin{align*}
P_{I\to F,i}^{(M)}
&=
\Pr\nolimits_{Q_i^{\bullet M}}
\left\{G_i^{(M)}\le 0\right\} \\
&=
\Pr\nolimits_{Q_i^{\bullet M}}
\left\{e^{-sG_i^{(M)}}\ge 1\right\} \\
&\le
\mathbb{E}_{Q_i^{\bullet M}}
\left[e^{-sG_i^{(M)}}\right].
\end{align*}
By independence,
\[
\mathbb{E}_{Q_i^{\bullet M}}
\left[e^{-sG_i^{(M)}}\right]
=
\left(
\mathbb{E}_{Q_i}
\left[e^{-s\Gamma_i(Y)}\right]
\right)^M .
\]
Again using the definition of $\Gamma_i(y)$, we have
\begin{align*}
\mathbb{E}_{Q_i}
\left[e^{-s\Gamma_i(Y)}\right]
&=
\int Q_i(y)
\left(\frac{P_i(y)}{Q_i(y)}\right)^s dy \\
&=
\int P_i(y)^sQ_i(y)^{1-s}\,dy .
\end{align*}
Using the same choice $s=\frac{1}{2}$ yields
\[
\mathbb{E}_{Q_i}
\left[e^{-\frac{1}{2}\Gamma_i(Y)}\right]
=
\int \sqrt{P_i(y)Q_i(y)}\,dy
=
\Omega_i ,
\]
and therefore
\[
P_{I\to F,i}^{(M)} \le \Omega_i^M .
\]

Finally, combining the two one-sided bounds with \eqref{eq:position_error_probability} yields
\[
P_{e,i}^{(M)}
=
\frac{1}{2}P_{F\to I,i}^{(M)}
+
\frac{1}{2}P_{I\to F,i}^{(M)}
\le
\Omega_i^M .
\]
This completes the proof.
\end{proof}

We now extend the position-wise bound to the sequence level under the same ideal SC-consistent model. In this extension, the recognition-consistent
prefix is used to condition on correct previous position-type decisions and to identify the first recognition error. The synthetic LLR at each position is
still generated under the correct-prefix condition, rather than from previous recognition decisions. Therefore, the resulting sequence-level bound
characterizes the intrinsic discrimination difficulty and does not account for the error-propagation effects of a practical SC-recursive implementation.

For this purpose, first note that, in the local binary test, $\mathcal{H}_{F,j}$ and $\mathcal{H}_{I,j}$ are mutually exclusive and exhaustive. 
Thus, the complement of the local error event $\mathcal{E}_j$ is the correct-decision event at position $j$, namely,
\begin{equation}
\mathcal{E}_j^c =
\bigl(\mathcal{H}_{F,j}\!\cap\!\{G_j^{(M)}\le 0\}\bigr)
\!\cup\!
\bigl(\mathcal{H}_{I,j}\!\cap\!\{G_j^{(M)}>0\}\bigr).
\end{equation}
Then, define the recognition-consistent prefix event as
\begin{equation}
\mathcal{R}_i
\triangleq
\bigcap_{j=0}^{i-1}\mathcal{E}_j^c,
\qquad i=0,1,\ldots,N-1,
\label{eq:recognition_prefix}
\end{equation}
where the empty intersection for $i=0$ is understood as the certain event.
Thus, $\mathcal{R}_i$ represents the event that all position-type decisions before index $i$ are correct in the ideal SC-consistent sequence-level analysis.
By construction of this model, conditioned on $\mathcal{R}_i$, the local test at index $i$ follows the same SC-consistent synthetic-channel model as that under $\mathcal{A}_i$.

Define
\begin{equation}
\mathcal{D}_i
\triangleq
\mathcal{R}_i\cap\mathcal{E}_i,
\qquad i=0,1,\ldots,N-1 .
\label{eq:first_error_event}
\end{equation}
Then $\mathcal{D}_i$ denotes the event that the first position-type recognition error occurs at index $i$. 
Hence, the overall sequence error event can be written as
\begin{equation}
\mathcal{E}_{\mathrm{tot}}
=
\bigcup_{i=0}^{N-1}\mathcal{D}_i .
\label{eq:overall_error_event}
\end{equation}
Since the events $\mathcal{D}_0,\mathcal{D}_1,\ldots,\mathcal{D}_{N-1}$ are mutually disjoint, we have the following proposition.

\begin{proposition}
\label{prop:overall_upper}
Under the neutral symmetric ideal SC-consistent sequence-level model, the overall recognition error probability satisfies
\begin{equation}
\Pr(\mathcal{E}_{\mathrm{tot}})
\le
\sum_{i=0}^{N-1}\Omega_i^M .
\label{eq:overall_upper_bound}
\end{equation}
\end{proposition}

\begin{proof}
Since the events $\mathcal{D}_i$ are mutually disjoint,
\[
\Pr(\mathcal{E}_{\mathrm{tot}})
=
\sum_{i=0}^{N-1}\Pr(\mathcal{D}_i).
\]
By the definition of $\mathcal{D}_i$,
\[
\Pr(\mathcal{D}_i)
=
\Pr(\mathcal{R}_i)\,
\Pr(\mathcal{E}_i\mid\mathcal{R}_i).
\]
Therefore,
\[
\Pr(\mathcal{E}_{\mathrm{tot}})
=
\sum_{i=0}^{N-1}
\Pr(\mathcal{R}_i)\,
\Pr(\mathcal{E}_i\mid\mathcal{R}_i).
\]

By construction, conditioned on \(\mathcal{R}_i\), the \(i\)-th local test is in the same neutral symmetric correct-prefix setting as the position-wise
test conditioned on \(\mathcal{A}_i\).
Therefore, the corresponding conditional error probability is bounded by Theorem~\ref{thm:upper_position}, i.e.,
\[
\Pr(\mathcal{E}_i\mid\mathcal{R}_i)
\le
\Omega_i^M .
\]
Since $\Pr(\mathcal{R}_i)\le 1$, it follows that
\[
\Pr(\mathcal{E}_{\mathrm{tot}})
\le
\sum_{i=0}^{N-1}\Omega_i^M .
\]
This completes the proof.
\end{proof}

The above proposition shows that each term \(\Omega_i^M\) in the sequence-level upper bound decays exponentially with \(M\) whenever
\(\Omega_i<1\). Thus, the bound is mainly controlled by the positions whose frozen- and information-position distributions have the largest overlap. 
For a practical SC-recursive procedure, an early erroneous position-type decision may further affect later positions. Section~V examines this effect.

\subsection{Hypothesis-Testing-Based Lower Bounds on Recognition Error Probability}

We next derive lower bounds under the same ideal SC-consistent hypothesis-testing model. The analysis is restricted to the neutral symmetric
case, where the two position-type hypotheses have equal local priors and equal misclassification costs. Since the neutral rule in Section~III is optimal in
this case, the bounds quantify the intrinsic difficulty of distinguishing frozen and information positions.

\begin{theorem}
\label{thm:lower_position}
Under \(\mathcal{A}_i\) and the neutral symmetric formulation, the conditional position-wise recognition error probability at position \(i\) satisfies

\begin{equation}
P_{e,i}^{(M)}
=
\frac{1}{2}
\left(
1-\left\|P_i^{\bullet M}-Q_i^{\bullet M}\right\|_{\mathrm{TV}}
\right),
\label{eq:lower_position_exact}
\end{equation}
where $\|\cdot\|_{\mathrm{TV}}$ denotes the total variation distance~\cite{VanTrees2001}.
Moreover,
\begin{equation}
P_{e,i}^{(M)}
\ge
\frac{1}{2}
\left(
1-\sqrt{1-\Omega_i^{2M}}
\right).
\label{eq:lower_position_bound}
\end{equation}
\end{theorem}

\begin{proof}
Under the neutral equal-prior and equal-cost formulation, the optimal error probability for discriminating between $P_i^{\bullet M}$ and $Q_i^{\bullet M}$ is given by \eqref{eq:lower_position_exact}.
Next, using the standard inequality
\[
\|P-Q\|_{\mathrm{TV}}
\le
\sqrt{1-B(P,Q)^2},
\]
together with the multiplicativity of the Bhattacharyya coefficient for product measures, we obtain
\begin{align*}
\left\|P_i^{\bullet M}-Q_i^{\bullet M}\right\|_{\mathrm{TV}}
&\le
\sqrt{
1-
B\!\left(P_i^{\bullet M},Q_i^{\bullet M}\right)^2
} \\
&=
\sqrt{1-\Omega_i^{2M}} .
\end{align*}
Substituting this bound into \eqref{eq:lower_position_exact} yields \eqref{eq:lower_position_bound}.
This completes the proof.
\end{proof}

We next extend the lower-bound analysis to the sequence level. Under the same ideal SC-consistent model, conditioned on \(\mathcal{R}_i\), the local test at
position \(i\) follows the neutral symmetric correct-prefix test considered in Theorem~3. This allows the position-wise lower bound to be used in the
following product decomposition.

Since
\[
\mathcal{E}_{\mathrm{tot}}^c
=
\bigcap_{i=0}^{N-1}\mathcal{E}_i^c ,
\]
the chain rule gives
\begin{equation}
\Pr(\mathcal{E}_{\mathrm{tot}}^c)
=
\prod_{i=0}^{N-1}
\Pr(\mathcal{E}_i^c\mid\mathcal{R}_i).
\label{eq:success_chain_rule}
\end{equation}
Since
\[
\Pr(\mathcal{E}_i^c\mid\mathcal{R}_i)
=
1-\Pr(\mathcal{E}_i\mid\mathcal{R}_i),
\]
we obtain the product representation
\begin{equation}
\Pr(\mathcal{E}_{\mathrm{tot}})
=
1-
\prod_{i=0}^{N-1}
\bigl(1-\Pr(\mathcal{E}_i\mid\mathcal{R}_i)\bigr).
\label{eq:overall_error_exact_product}
\end{equation}

\begin{proposition}
\label{prop:overall_lower}
Under the neutral symmetric ideal SC-consistent sequence-level formulation, the overall recognition error probability satisfies
\begin{equation}
\Pr(\mathcal{E}_{\mathrm{tot}})
\ge
1-
\prod_{i=0}^{N-1}
\left(
1-
\frac{1}{2}
\left(
1-\sqrt{1-\Omega_i^{2M}}
\right)
\right).
\label{eq:overall_lower_bound}
\end{equation}
\end{proposition}

\begin{proof}

Given \(\mathcal{R}_i\), the local test at position \(i\) has the same neutral symmetric correct-prefix model as the position-wise test considered in Theorem~\ref{thm:lower_position}. Hence,
\[
\Pr(\mathcal{E}_i\mid\mathcal{R}_i)
=
P_{e,i}^{(M)} .
\]
By Theorem~\ref{thm:lower_position}, for each $i$,
\[
P_{e,i}^{(M)}
\ge
\frac{1}{2}
\left(
1-\sqrt{1-\Omega_i^{2M}}
\right).
\]
Therefore,
\[
\Pr(\mathcal{E}_i\mid\mathcal{R}_i)
\ge
\frac{1}{2}
\left(
1-\sqrt{1-\Omega_i^{2M}}
\right).
\]

Now consider the function
\[
f(x_0,\ldots,x_{N-1})
=
1-\prod_{i=0}^{N-1}(1-x_i),
\qquad
0\le x_i\le 1.
\]
This function is monotonically increasing with respect to each argument $x_i$.
Therefore, by \eqref{eq:overall_error_exact_product},
\begin{align*}
\Pr(\mathcal{E}_{\mathrm{tot}})
&=
1-
\prod_{i=0}^{N-1}
\bigl(1-\Pr(\mathcal{E}_i\mid\mathcal{R}_i)\bigr) \\
&\ge
1-
\prod_{i=0}^{N-1}
\left(
1-
\frac{1}{2}
\left(
1-\sqrt{1-\Omega_i^{2M}}
\right)
\right).
\end{align*}
This completes the proof.
\end{proof}

Thus, the lower bound characterizes the intrinsic difficulty of frozen/information
position discrimination under the neutral symmetric ideal SC-consistent model.

\subsection{Connection to Classical Bhattacharyya Parameters}

We now relate the position-type overlap coefficient $\Omega_i$ to the classical Bhattacharyya parameter of the $i$-th synthetic bit-channel.
Specifically, define
\begin{equation}
\begin{split}
Z_i
&\triangleq
Z(W_i) \\
&\triangleq
B\!\bigl(W_i(\cdot|0),W_i(\cdot|1)\bigr) \\
&=
\int \sqrt{W_i(y|0)W_i(y|1)}\,dy .
\end{split}
\label{eq:classical_bhattacharyya}
\end{equation}
This connection is important because $Z_i$ is the standard reliability measure used in polar-code analysis, whereas $\Omega_i$ is the key quantity governing blind recognition in the present SC-consistent framework.

Let $Y\sim P_i$ and define
\[
X_i(Y)
\triangleq
\sqrt{\frac{W_i(Y|1)}{W_i(Y|0)}} .
\]
Then
\[
\mathbb{E}_{P_i}[X_i]=Z_i,
\qquad
\mathbb{E}_{P_i}[X_i^2]=1 .
\]
On the other hand, from the definitions of $P_i$ and $Q_i$,
\[
\Omega_i
=
\mathbb{E}_{P_i}\!\left[\phi(X_i)\right],
\qquad
\phi(x)\triangleq\sqrt{\frac{1+x^2}{2}},
\quad x\ge 0 .
\]

The following theorem provides lower and upper bounds on $\Omega_i$ in terms of $Z_i$ only.

\begin{theorem}
\label{thm:Omega_Z_relation}
For each source position $i$,
\begin{equation}
\sqrt{\frac{1+Z_i^2}{2}}
\le
\Omega_i
\le
\psi(Z_i),
\label{eq:Omega_Z_relation}
\end{equation}
where
\begin{equation}
\psi(z)
\triangleq
\frac{1-z^2}{\sqrt{2}}
+
z\sqrt{\frac{1+z^2}{2}},
\qquad 0\le z\le 1 .
\label{eq:psi_definition}
\end{equation}
\end{theorem}

\begin{proof}
We first prove the lower bound.
Since $\phi(x)=\sqrt{(1+x^2)/2}$ is convex on $[0,\infty)$, Jensen's inequality gives
\[
\Omega_i
=
\mathbb{E}_{P_i}\!\left[\phi(X_i)\right]
\ge
\phi\!\left(\mathbb{E}_{P_i}[X_i]\right)
=
\phi(Z_i)
=
\sqrt{\frac{1+Z_i^2}{2}} .
\]

We next prove the upper bound.
We first consider the case $Z_i=0$.
Since $X_i\ge 0$ and $\mathbb{E}_{P_i}[X_i]=Z_i=0$, we have $X_i=0$ $P_i$-almost surely.
Therefore,
\[
\Omega_i
=
\mathbb{E}_{P_i}\!\left[\phi(X_i)\right]
=
\phi(0)
=
\frac{1}{\sqrt{2}}
=
\psi(0),
\]
which proves the desired upper bound for $Z_i=0$.

It remains to consider the case $Z_i>0$.
Fix $z\in(0,1]$ and define the quadratic polynomial
\[
q_z(x)=a_zx^2+b_zx+\frac{1}{\sqrt{2}},
\]
where
\begin{align*}
a_z
&=
\frac{z^2}{\sqrt{2}}
\left(1-\frac{z}{\sqrt{1+z^2}}\right), \\
b_z
&=
\frac{1}{\sqrt{2}}
\left(\frac{2z^2+1}{\sqrt{1+z^2}}-2z\right).
\end{align*}
These coefficients are chosen so that
\begin{align*}
q_z(0) &= \phi(0),\\
q_z(1/z) &= \phi(1/z),\\
q_z'(1/z) &= \phi'(1/z).
\end{align*}

Now, the third derivative of $\phi(x)$ is
\[
\phi'''(x)
=
-\frac{3x}{\sqrt{2}(1+x^2)^{5/2}},
\qquad x\ge 0,
\]
and hence $\phi'''(x)\le 0$ for all $x\ge 0$.
By the Hermite interpolation remainder formula, for every $x\ge 0$, there exists some $\xi_x\in(0,\max\{x,1/z\})$ such that
\[
\phi(x)-q_z(x)
=
\frac{\phi'''(\xi_x)}{6}\,
x\left(x-\frac{1}{z}\right)^2 .
\]
Since $x\ge 0$, $\left(x-\frac{1}{z}\right)^2\ge 0$, and $\phi'''(\xi_x)\le 0$, it follows that
\[
\phi(x)-q_z(x)\le 0,
\qquad \forall x\ge 0 .
\]
Therefore, $\phi(x)\le q_z(x)$ for all $x\ge 0$.

Applying this inequality to $X_i$ with $z=Z_i>0$ and taking expectation under $P_i$, we obtain
\begin{align*}
\Omega_i
&=
\mathbb{E}_{P_i}\!\left[\phi(X_i)\right] \\
&\le
\mathbb{E}_{P_i}\!\left[q_{Z_i}(X_i)\right] \\
&=
a_{Z_i}\mathbb{E}_{P_i}[X_i^2]
+
b_{Z_i}\mathbb{E}_{P_i}[X_i]
+
\frac{1}{\sqrt{2}} \\
&=
a_{Z_i}+b_{Z_i}Z_i+\frac{1}{\sqrt{2}} \\
&=
\psi(Z_i).
\end{align*}
This completes the proof.
\end{proof}

By combining Theorem~\ref{thm:Omega_Z_relation} with the upper and lower bounds derived in Sections~IV-A and IV-B, we obtain recognition bounds expressed solely in terms of the classical Bhattacharyya parameter $Z_i$.

\begin{corollary}
\label{cor:Z_only_bounds}
Under the neutral symmetric ideal SC-consistent model, for each source position \(i\), the position-wise recognition error probability satisfies
\begin{equation}
\frac{1}{2}
\left(
1-
\sqrt{
1-
\left(\frac{1+Z_i^2}{2}\right)^M
}
\right)
\le
P_{e,i}^{(M)}
\le
\psi(Z_i)^M .
\label{eq:Z_only_position_bound}
\end{equation}
Moreover, the sequence-level recognition error probability satisfies the following upper and lower bounds:
\begin{equation}
\Pr(\mathcal{E}_{\mathrm{tot}})
\le
\sum_{i=0}^{N-1}\psi(Z_i)^M ,
\label{eq:Z_only_overall_upper}
\end{equation}
and
\begin{equation}
\Pr(\mathcal{E}_{\mathrm{tot}})
\ge
1-
\prod_{i=0}^{N-1}
\Bigl(
1-\frac{1}{2}
\bigl(
1-\sqrt{1-\bigl(\tfrac{1+Z_i^2}{2}\bigr)^M}
\bigr)
\Bigr).
\label{eq:Z_only_overall_lower}
\end{equation}
\end{corollary}

\begin{proof}
From Theorem~\ref{thm:upper_position} and the upper bound in Theorem~\ref{thm:Omega_Z_relation},
\[
P_{e,i}^{(M)}
\le
\Omega_i^M
\le
\psi(Z_i)^M .
\]
Similarly, by Theorem~\ref{thm:lower_position} and the lower bound in Theorem~\ref{thm:Omega_Z_relation},
\begin{align*}
P_{e,i}^{(M)}
&\ge
\frac{1}{2}\bigl(1-\sqrt{1-\Omega_i^{2M}}\bigr) \\
&\ge
\frac{1}{2}
\Bigl(
1-\sqrt{1-\bigl(\tfrac{1+Z_i^2}{2}\bigr)^M}
\Bigr).
\end{align*}

The overall upper bound follows from Proposition~\ref{prop:overall_upper} together with the inequality $\Omega_i^M\le \psi(Z_i)^M$.
The overall lower bound follows from Proposition~\ref{prop:overall_lower} together with the position-wise lower bound established above.
This completes the proof.
\end{proof}

Corollary~2 expresses the \(\Omega_i\)-based recognition bounds in terms of the classical Bhattacharyya parameter \(Z_i\). This connects the recognition analysis to standard polar-code reliability measures. The resulting \(Z_i\)-based bounds are generally looser because \(Z_i\) summarizes synthetic-channel reliability rather than the full overlap
between the frozen- and information-position models.

\section{Simulation Results}

In this section, simulation results are presented to evaluate the proposed hypothesis-testing-based analysis for blind recognition of polar codes.
All polar codes are constructed by the Gaussian approximation (GA) method~\cite{GA} with a design value of $E_b/N_0=2$ dB, modulated using BPSK, and transmitted over the AWGN channel.
The design value is used only for polar-code construction, whereas the transmission SNR in the performance curves is represented by $E_s/N_0$.
Here, $C(N,K)$ specifies the actual polar code used to generate the simulated transmissions and to evaluate recognition errors.
During recognition, only the code length $N$ is assumed known, whereas both the dimension $K$ and the frozen/information position pattern are unknown, reflecting the non-cooperative setting considered in this work.

For each source position, the recognition statistic is accumulated over $M$ independent intercepted observations and compared with the neutral threshold $\ln 2$, corresponding to $\tau_i=0$ in Section~III-B.
A sequence-level error is declared if at least one position is incorrectly identified.
We consider two simulation procedures.
In the ideal SC-consistent procedure, the synthetic LLRs are evaluated under the correct-prefix condition used in the theoretical analysis.
In the realistic SC-recursive procedure, previous recognition and hard-decision results are fed back into subsequent SC operations.
The ideal procedure is used to evaluate the bounds in Section~IV,  whereas the realistic procedure is included to show the effect of recursive feedback.
The classical Bhattacharyya parameters $Z_i$ used in the theoretical bounds are computed by numerically integrating the corresponding synthetic-channel densities under the AWGN channel at each simulated $E_s/N_0$.

\begin{figure}[t]
  \centering
  \includegraphics[width=1\linewidth]{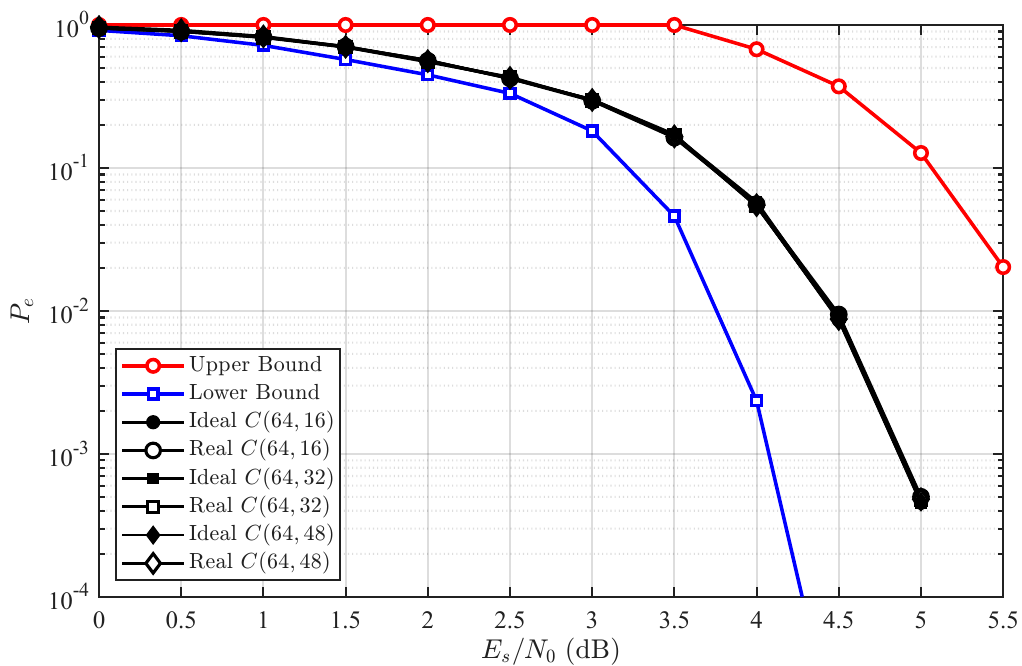}
  \caption{Performance comparison with $N=64,M=100$.}
  \label{fig:64M100}
\end{figure}

\begin{figure}[t]
  \centering
  \includegraphics[width=1\linewidth]{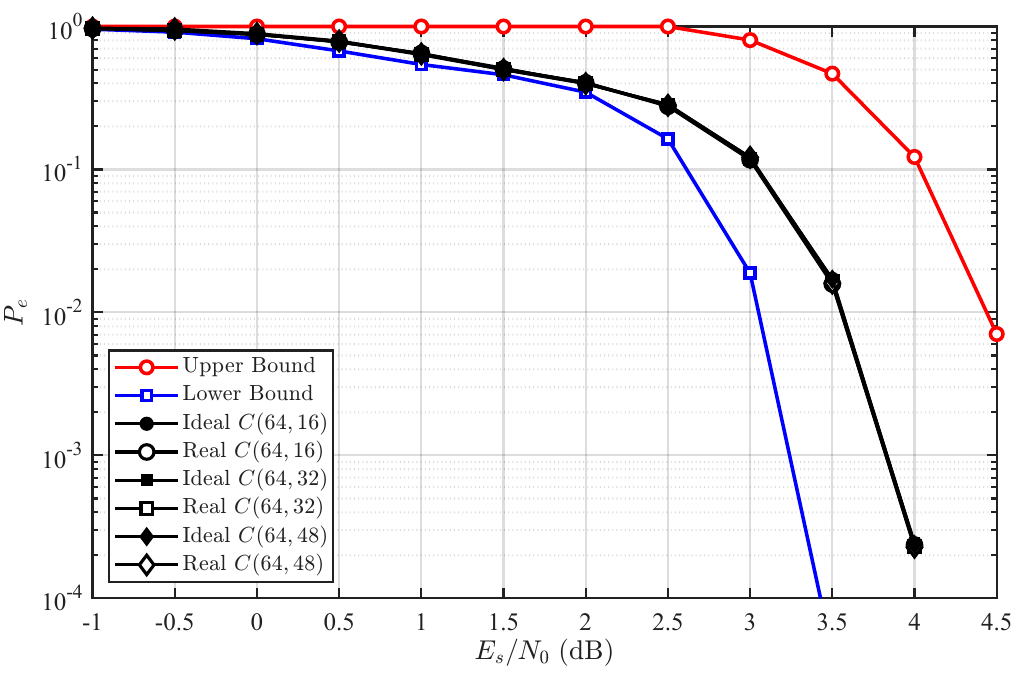}
  \caption{Performance comparison with $N=64,M=500$.}
  \label{fig:64M500}
\end{figure}

\begin{figure}[t]
  \centering
  \includegraphics[width=1\linewidth]{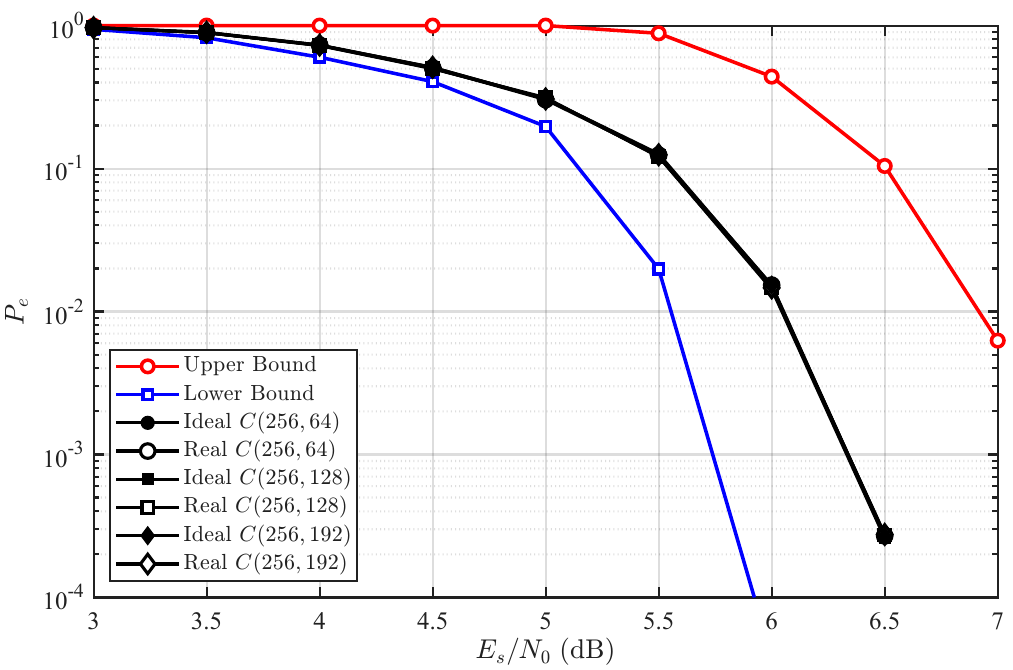}
  \caption{Performance comparison with $N=256,M=100$.}
  \label{fig:256M100}
\end{figure}

\begin{figure}[t]
  \centering
  \includegraphics[width=1\linewidth]{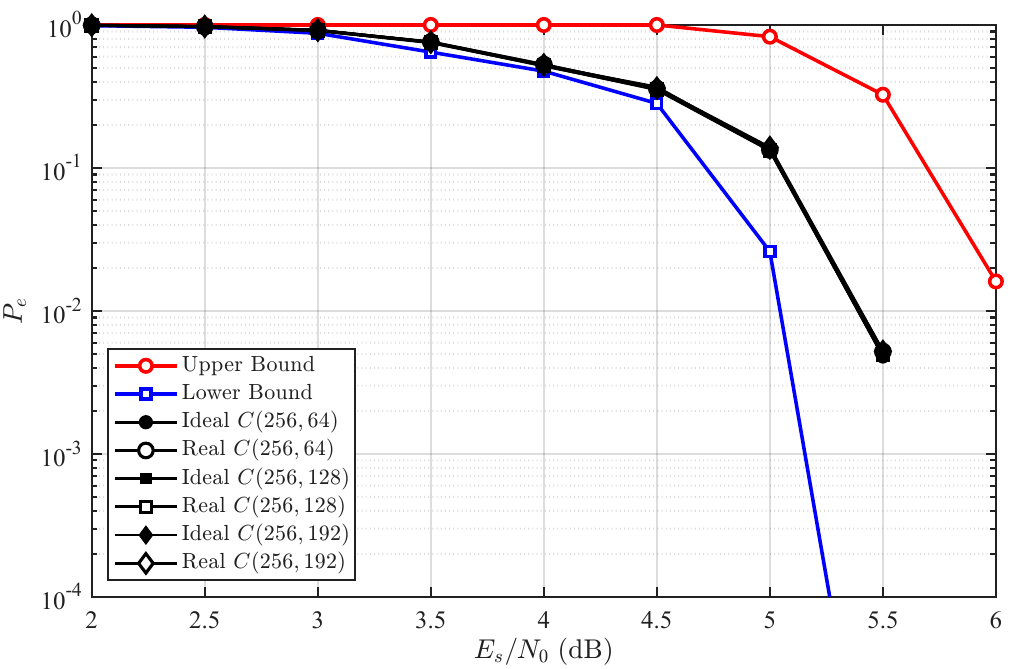}
  \caption{Performance comparison with $N=256,M=500$.}
  \label{fig:256M500}
\end{figure}

\subsection{Comparison with the Theoretical Bounds}

Figs.~\ref{fig:64M100}--\ref{fig:256M500} compare the sequence-level recognition error probability with the theoretical bounds in~\eqref{eq:Z_only_overall_upper} and~\eqref{eq:Z_only_overall_lower}.
The curves labeled ``Ideal'' represents the  ideal SC-consistent procedure.
The curves labeled ``Real'' represents the realistic SC-recursive procedure.

The comparisons are made for $N=64$ and $N=256$, with $M=100$ and $M=500$ intercepted observations.
For $N=64$, the transmitted polar codes are $C(64,16)$, $C(64,32)$, and $C(64,48)$.
For $N=256$, the transmitted polar codes are $C(256,64)$, $C(256,128)$, and $C(256,192)$.
The horizontal axis denotes $E_s/N_0$ in dB, and the vertical axis denotes the sequence-level recognition error probability.
The value of $K$ is used only to generate the transmitted polar code and is not provided to the recognition procedure.

As $E_s/N_0$ increases, all curves show a decrease in the sequence-level error probability.
From the hypothesis-testing viewpoint, a higher SNR increases the separation between the frozen-position and information-position observation laws, which reduces the probability of an incorrect local decision.

The number of intercepted observations affects the transition region of the curves.
For $N=64$, increasing $M$ from $100$ to $500$ moves the error-probability curve to a lower SNR range.
The sequence-level error probability reaches the order of $10^{-2}$ at approximately $4.5$ dB for $M=100$, while the same level is reached at approximately $3.55$ dB for $M=500$.
The same leftward shift is also observed for $N=256$.
This trend agrees with the bounds in Section~IV, where the relevant error terms decrease with the number of independent observations.

The effect of the code length is mainly reflected in the sequence-level criterion.
For a fixed $M$, the $N=256$ cases require a higher SNR than the corresponding $N=64$ cases to reach the same error level.
For example, when $M=100$, the transition to the $10^{-2}$ error-probability level occurs at approximately $4.5$ dB for $N=64$ and approximately $6.05$ dB for $N=256$.
With a longer code, more source positions must be identified, making the sequence-level error probability more sensitive to poorly distinguishable positions.

The ideal and realistic curves are close at the sequence level in the tested cases.
The remaining difference comes from the feedback in the realistic SC-recursive procedure, where an early wrong decision can affect subsequent synthetic LLRs and cause additional position errors.
This effect is examined in more detail by the paired diagnostic statistics in Section~V-B.
For fixed $N$, $M$, and SNR, the curves with different values of $K$ are also close, indicating that the sequence-level performance in these tests is not strongly affected by the transmitted code rate.

The ideal simulation results are consistent with the range indicated by the theoretical bounds.
The upper bound is conservative because it combines position-wise bounds and the Bhattacharyya-parameter relaxation, while the lower bound gives a reference for the intrinsic sequence-level difficulty under the SC-consistent model.

\begin{table*}[t]
\centering
\caption{Statistics for comparing the SC-consistent and realistic recursions with $N=64$, $M=100$.}
\label{tab:ideal_real_statistics}
\scriptsize
\setlength{\tabcolsep}{4.2pt}
\renewcommand{\arraystretch}{1.10}
\resizebox{\textwidth}{!}{%
\begin{tabular}{c c c c c c c c c c c}
\toprule
Code 
& SNR
& $P_{\rm seq}$
& $\widetilde P_{\rm seq}$
& $p_{00}$
& $p_{11}$
& $p_{10}$
& $p_{01}$
& $P_{\rm pos}$
& $\widetilde P_{\rm pos}$
& $P_{\rm FS}$  \\
\midrule

\multirow{6}{*}{$C(64,16)$}
& 0 dB & 0.9544 & 0.9544 & 0.0456 & 0.9544 & 0 & 0 & 0.03883 & 0.6262 & 1 \\
& 1 dB & 0.8212 & 0.8212 & 0.1788 & 0.8212 & 0 & 0 & 0.02257 & 0.5238 & 1 \\
& 2 dB & 0.5603 & 0.5603 & 0.4397 & 0.5603 & 0 & 0 & 0.01055 & 0.3357 & 1 \\
& 3 dB & 0.2962 & 0.2962 & 0.7038 & 0.2962 & 0 & 0 & 0.004680 & 0.1733 & 1 \\
& 4 dB & 0.05671 & 0.05671 & 0.9433 & 0.05671 & 0 & 0 & 0.0008861 & 0.03659 & 1 \\
& 5 dB & 0.0004745 & 0.0004745 & 0.9995 & 0.0004745 & 0 & 0 & 0.000007414 & 0.0003250 & 1 \\

\midrule

\multirow{6}{*}{$C(64,32)$}
& 0 dB & 0.9537 & 0.9537 & 0.0463 & 0.9537 & 0.00005 & 0 & 0.03858 & 0.3994 & 1 \\
& 1 dB & 0.8268 & 0.8268 & 0.1732 & 0.8268 & 0 & 0 & 0.02265 & 0.3424 & 1 \\
& 2 dB & 0.5598 & 0.5598 & 0.4403 & 0.5598 & 0 & 0 & 0.01064 & 0.2247 & 1 \\
& 3 dB & 0.2930 & 0.2930 & 0.7071 & 0.2930 & 0 & 0 & 0.004637 & 0.1169 & 1 \\
& 4 dB & 0.05679 & 0.05679 & 0.9432 & 0.05679 & 0 & 0 & 0.0008873 & 0.02481 & 1 \\
& 5 dB & 0.0004625 & 0.0004625 & 0.9995 & 0.0004625 & 0 & 0 & 0.000007227 & 0.0002172 & 1 \\

\midrule

\multirow{6}{*}{$C(64,48)$}
& 0 dB & 0.9520 & 0.9665 & 0.03125 & 0.9497 & 0.01675 & 0.00230 & 0.03864 & 0.1835 & 0.9970 \\
& 1 dB & 0.8295 & 0.8329 & 0.1657 & 0.8280 & 0.00485 & 0.00145 & 0.02259 & 0.1637 & 0.9999 \\
& 2 dB & 0.5653 & 0.5653 & 0.4347 & 0.5652 & 0.00010 & 0.00005 & 0.01081 & 0.1137 & 1 \\
& 3 dB & 0.2980 & 0.2980 & 0.7021 & 0.2980 & 0 & 0 & 0.004723 & 0.06035 & 1 \\
& 4 dB & 0.05721 & 0.05721 & 0.9428 & 0.05721 & 0 & 0 & 0.0008939 & 0.01266 & 1 \\
& 5 dB & 0.0004750 & 0.0004750 & 0.9995 & 0.0004750 & 0 & 0 & 0.000007422 & 0.0001119 & 1 \\

\bottomrule
\end{tabular}%
}
\end{table*}

\subsection{Diagnostic Statistics for the Ideal and Realistic Recursions}

Table~\ref{tab:ideal_real_statistics} reports paired diagnostic statistics for $N=64$ and $M=100$.
In each Monte Carlo trial, the SC-consistent and realistic SC-recursive procedures are applied to the same received channel-LLR matrix $\Lambda^{\rm ch}$.
Thus, the differences in the table are caused by the recursion mechanism rather than by different channel-noise realizations.

In the table, $P_{\rm seq}$ and $\widetilde P_{\rm seq}$ denote the sequence-level error probabilities of the SC-consistent and realistic procedures, respectively.
The quantities $p_{00}$, $p_{11}$, $p_{10}$, and $p_{01}$ classify the paired sequence-level outcomes.
The first subscript refers to the realistic procedure and the second to the SC-consistent procedure, with $0$ denoting success and $1$ denoting failure.
Thus, $p_{00}$ and $p_{11}$ correspond to matched outcomes, whereas $p_{10}$ and $p_{01}$ correspond to mismatch cases.

The sequence-level statistics show that the two procedures give almost the same failure events in most tested cases.
For $C(64,16)$, $P_{\rm seq}$ and $\widetilde P_{\rm seq}$ are identical at all listed SNR points, and $p_{10}=p_{01}=0$.
For $C(64,32)$, a mismatch appears only at $0$ dB, where $p_{10}=5\times 10^{-5}$ and $p_{01}=0$.
For $C(64,48)$, the mismatch is more visible at low SNR, with $p_{10}=0.01675$ and $p_{01}=0.00230$ at $0$ dB, but both mismatch probabilities become zero from $3$ dB onward.
These values are consistent with the close ideal and realistic curves observed in Figs.~\ref{fig:64M100}--\ref{fig:64M500}.

The table also shows that the sequence-level error probability is only weakly affected by the transmitted code rate in this set of simulations.
For example, at $2$ dB, the SC-consistent sequence-level error probabilities are $0.5603$, $0.5598$, and $0.5653$ for $C(64,16)$, $C(64,32)$, and $C(64,48)$, respectively.
The corresponding realistic values are also close.
This agrees with the rate-insensitive behavior observed in Figs.~\ref{fig:64M100}--\ref{fig:64M500}.

The position-level statistics show a different result.
The quantities $P_{\rm pos}$ and $\widetilde P_{\rm pos}$ denote the normalized position-level error rates of the SC-consistent and realistic procedures, respectively.
At $2$ dB, $P_{\rm pos}$ remains around $1.1\times 10^{-2}$ for all three code rates, whereas $\widetilde P_{\rm pos}$ is $0.3357$, $0.2247$, and $0.1137$ for $C(64,16)$, $C(64,32)$, and $C(64,48)$, respectively.
Thus, even when the sequence-level error probabilities are nearly the same, the realistic recursion can have a much larger normalized position-level error rate.

This position-level gap comes from feedback in the realistic SC-recursive procedure.
In the SC-consistent procedure, each position is evaluated under the correct-prefix condition, so a wrong decision at one position does not change the synthetic LLRs used at later positions.
In the realistic procedure, previous recognition and hard-decision results are fed back into subsequent SC operations.
An early wrong decision may therefore distort later synthetic LLRs, especially through recursive $g$-operations, and trigger additional position errors.

The first-error statistics clarify why this position-level discrepancy does not lead to a comparable sequence-level discrepancy.
Here, $P_{\rm FS}$ denotes the probability that the two procedures have the same first erroneous position, conditioned on the event that both procedures fail.
For $C(64,16)$ and $C(64,32)$, $P_{\rm FS}=1$ at all listed SNR points.
For $C(64,48)$, $P_{\rm FS}$ is close to one at low SNR and becomes one from $2$ dB onward.
Hence, in the common failure cases under the same received LLRs, the two procedures usually first fail at the same source position.
After this first error, the realistic recursion may accumulate additional position errors, but the sequence-level failure event has already occurred.
As a result, the two procedures can have close sequence-level error probabilities while showing very different position-level error rates.

\section{Conclusion}
This paper provided a hypothesis-testing interpretation of soft-metric-based blind recognition for polar codes.
Under the SC-consistent synthetic-channel model, the frozen-position and information-position hypotheses lead to two different observation laws, and the adopted soft recognition metric is shown to be a shifted log-likelihood ratio between them.
This result explains why $\ln 2$ serves as the neutral threshold for the averaged metric, while unequal priors or costs correspond to a threshold shift.

Under the neutral rule, the analysis with a finite number of intercepted observations yields both position-wise and sequence-level bounds.
The bounds are determined by the overlap coefficient between the frozen-position and information-position models, which is further linked to the Bhattacharyya parameter of the corresponding synthetic channel.
The simulations show that the recognition performance improves with more intercepted observations and higher SNR, but becomes more difficult at the sequence level when the code length increases.
The paired tests show that the SC-consistent and realistic SC-recursive procedures usually fail first at the same position, leading to close sequence-level performance in the tested settings.


\begin{thebibliography}{99}

\bibitem{Dobre2007}
M.~Z. Hameed, A. Gy\"orgy, and D. G\"und\"uz, 
``The best defense is a good offense: Adversarial attacks to avoid modulation detection,'' 
\emph{IEEE Trans. Inf. Forensics Security}, 
vol.~16, pp.~1074--1087, 2021.


\bibitem{Shiu2011}
N. Xie, Z. Li, and H.~J. Tan, 
``A survey of physical-layer authentication in wireless communications,'' 
\emph{IEEE Commun. Surveys Tuts.}, 
vol.~23, no.~1, pp.~282--310, First Quart. 2021.

\bibitem{MoosaviLarsson2011}
S. Ramabadran, A.~S. Madhu Kumar, W. Guohua, and T.~S. Kee, 
``Blind recognition of LDPC code parameters over erroneous channel conditions,'' 
\emph{IET Signal Process.}, 
vol.~13, no.~1, pp.~86--95, Feb. 2019.

\bibitem{Xu2011}
W. Wang, Z. Sun, K. Ren, B. Zhu, and S. Piao, 
``Wireless physical-layer identification: Modeling and validation,'' 
\emph{IEEE Trans. Inf. Forensics Security}, 
vol.~11, no.~9, pp.~2091--2106, Sept. 2016.

\bibitem{MoosaviLarsson2014}
R. Moosavi and E. G. Larsson, ``Fast blind recognition of channel codes,'' \emph{IEEE Trans. Commun.}, vol. 62, no. 5, pp. 1393--1405, May 2014.

\bibitem{Dehdashtian2021}
S. Dehdashtian, M. Hashemi, and S. Salehkaleybar, ``Deep-learning-based blind recognition of channel code parameters over candidate sets under AWGN and multi-path fading conditions,'' \emph{IEEE Wireless Commun. Lett.}, vol. 10, no. 5, pp. 1041--1045, May 2021.

\bibitem{Mukherjee2014}
A. Mukherjee, S. A. A. Fakoorian, J. Huang, and A. L. Swindlehurst, ``Principles of physical layer security in multiuser wireless networks: A survey,'' \emph{IEEE Commun. Surveys Tuts.}, vol. 16, no. 3, pp. 1550--1573, 3rd Quart. 2014.

\bibitem{PoorSchaefer2017}
H. V. Poor and R. F. Schaefer, ``Wireless physical layer security,'' \emph{Proc. Natl. Acad. Sci. U.S.A.}, vol. 114, no. 1, pp. 19--26, Jan. 2017.


\bibitem{Arikan2009}
E. Ar{\i}kan, ``Channel polarization: A method for constructing capacity-achieving codes for symmetric binary-input memoryless channels,'' \emph{IEEE Trans. Inf. Theory}, vol. 55, no. 7, pp. 3051--3073, Jul. 2009.


\bibitem{TalVardy2015}
I. Tal and A. Vardy, ``List decoding of polar codes,'' \emph{IEEE Trans. Inf. Theory}, vol. 61, no. 5, pp. 2213--2226, May 2015.

\bibitem{NiuChen2012}
K. Niu and K. Chen, ``CRC-aided decoding of polar codes,'' \emph{IEEE Commun. Lett.}, vol. 16, no. 10, pp. 1668--1671, Oct. 2012.

\bibitem{TS38212}
3rd Generation Partnership Project (3GPP), ``NR; Multiplexing and channel coding,'' TS 38.212, Rel. 15, Jul. 2018.

\bibitem{Condo2017}
C. Condo, S. A. Hashemi, and W. J. Gross, ``Blind detection with polar codes,'' \emph{IEEE Commun. Lett.}, vol. 21, no. 12, pp. 2550--2553, Dec. 2017.

\bibitem{Liu2022}
J. Liu, T. Zhang, H. Bai, and S. Ye, ``Blind recognition algorithm of polar code based on information matrix estimation,'' \emph{Syst. Eng. Electron.}, vol. 44, no. 2, pp. 668--676, 2022.

\bibitem{Yi2023}
C. Yi, B. Pang, L. He, B. Ma, Y. Li, and F. C. M. Lau, ``Blind identification of polar codes based on estimation and derivation approaches,'' \emph{IEEE Commun. Lett.}, vol. 27, no. 2, pp. 414--418, Feb. 2023.

\bibitem{Xu2026}
P. Xu, J. Liu, A. Wang, C. Yi, and Q. Li, ``Blind recognition of polar code information bits based on multi-threshold voting and partial orders,'' \emph{IEEE Commun. Lett.}, vol. 30, pp. 887--891, 2026.

\bibitem{Wang2023}
Y. Wang, C. Wang, X. Wang, and Z. Huang, ``Non-punctured polar code parameter recognition algorithm based on soft decision,'' \emph{Syst. Eng. Electron.}, vol. 45, no. 10, pp. 3293--3301, Oct. 2023.


\bibitem{TuBSCL}
C. Tu, Y. Liu, X. Feng, and K. Niu, ``Blind recognition of polar codes using successive cancellation list decoding,'' \emph{arXiv preprint arXiv:2605.13331}, May 2026.


\bibitem{VanTrees2001}
H. L. Van Trees, \emph{Detection, Estimation, and Modulation Theory, Part I}. New York, NY, USA: Wiley, 2001.

\bibitem{CoverThomas2006}
T. M. Cover and J. A. Thomas, \emph{Elements of Information Theory}, 2nd ed. Hoboken, NJ, USA: Wiley-Interscience, 2006.

\bibitem{GA}	
P. Trifonov, “Efficient design and decoding of polar codes,” \emph{IEEE Trans. Commun.}, vol. 60, no. 11, pp. 3221–3227, Nov. 2012.

\end{thebibliography}
\end{document}